\documentclass[12pt]{article}
\usepackage[english]{babel}
\usepackage{mathrsfs}
\usepackage{amssymb}
\usepackage{hyperref}
\usepackage{amssymb,amsmath,amsfonts}
\newtheorem{theorem}{Theorem}

\newtheorem{corollary}[theorem]{Corollary}
\newtheorem{proposition}[theorem]{Proposition}

\begin{document}

\title{ On the number of $q$-ary  quasi-perfect codes with covering radius 2}

\date{}

\author{Alexander M. Romanov
\thanks {Sobolev Institute of Mathematics, 630090 Novosibirsk, Russia. 
Supported by the program of fundamental scientific researches of the SB RAS I.5.1, project N 0314-2019-0017. }
%\small Sobolev Institute of Mathematics\\[-0.8ex]
%\small  630090 Novosibirsk, Russia\\
%\small\tt {rom@math.nsc.ru}
}

\maketitle

\begin{abstract}
In this paper we present a family of $q$-ary   nonlinear quasi-perfect codes  with covering radius 2. The
codes have   length $n = q^m$ and size $ M = q^{n - m - 1}$ where $q$ is a prime power, $q \geq 3$,  $m$ is an integer, $m \geq 2$.
We  prove that there are more than   $q^{q^{cn}}$ nonequivalent such codes of length $n$, for all sufficiently large $n$ and a constant
$c = \frac{1}{q} - \varepsilon$.
\end{abstract}

\section{Introduction}\label{sec:intr}

%%%%%%%%%%%%%%%%%%%%%%%%%%%%%%%%%%%%%%%%%%%%%%%%%%%

Let $\mathbb{F}_{q}^n$ be a vector space of dimension $n$ over the finite field $\mathbb{F}_{q}$,  where $q$ is a prime power.
An arbitrary subset  $ {\mathcal C}$ of  $\mathbb{F}_{q}^n$ is called a \emph{$q$-ary  error correcting  code}
(briefly a $q$-ary code).
The \emph{length} of a code ${\mathcal C} \subseteq \mathbb{F}_{q}^n$ is  the dimension of the vector space  $\mathbb{F}_{q}^n$.
We assume that the all-zero  vector ${\bf 0}$ always belongs to the code, unless otherwise specified.
A code is called \emph{linear} if it  is a subspace   of  $\mathbb{F}_{q}^n$.
Otherwise, the code is called \emph{nonlinear}.
The \emph{dimension} of a  linear code   is  the dimension of the  subspace, denoted by $\text{dim}({\cal C})$.
The vectors belonging to a code ${\cal C}$, we will call \emph{codewords}.
The \emph{Hamming weight} of a vector  ${\bf x}  \in \mathbb{F}_{q}^{n}$ is   the number of nonzero coordinate positions of ${\bf x}$, denoted by $wt(\bf x)$.
The \emph{Hamming distance} between two vectors ${\bf x}$, ${\bf y} \in \mathbb{F}_{q}^n$ is the number of coordinate positions in which they differ, denoted by $d({\bf x}, {\bf y})$.
The \emph{minimum distance} of a   code $\mathcal C$  is  the smallest Hamming distance between two different codewords of $\mathcal C$, denoted by $d(\mathcal C)$.

Define the \emph{packing radius} $e({\mathcal C})$ of a code ${\cal C} \subseteq  \mathbb{F}_{q}^{n}$.
Let  $e({\mathcal C}) = {\lfloor{(d - 1)}/{2}\rfloor}$, where $d$ is the minimum distance of ${\mathcal C}$.
We also define the \emph{covering radius} $\rho({\mathcal C})$ of a code ${\mathcal C}$.
Let
\begin{equation*}
\rho({\mathcal C}) =  \max\limits_{{\bf x} \in \mathbb{F}_{q}^{n}}\min\limits_{{\bf c} \in {\mathcal C}}d({\bf x},{\bf c}).
\end{equation*}

A code ${\cal C} \subseteq  \mathbb{F}_{q}^{n}$ is called \emph{perfect} if  $\rho({\mathcal C}) = e({\mathcal C})$.
If  $\rho({\mathcal C}) = e({\mathcal C}) + 1$, then  ${\mathcal C}$ is called \emph{quasi-perfect}.  See  \cite[p. 19]{mac}.

Two codes ${\mathcal C}_1, {\mathcal C}_2 \subseteq \mathbb{F}_{q}^n $ are said to be \emph{equivalent} if there exists a vector
${\bf v} \in \mathbb{F}_{q}^n $ and a $n \times n$ monomial matrix $M$ over $\mathbb{F}_{q}$ such that
$${\mathcal C}_2 =\{({\bf v} + {\bf c}M) \, \,  |  \, \,  {\bf c} \in  {\mathcal C}_1\}.$$

A codeword ${\bf c}  = (c_1, c_2, \ldots, c_n)$   is called \emph{even} if $ \sum_{i = 1}^n c_i = 0$.
A code is called \emph{even} if it contains only even codewords.

We will use the standard notation  $(n, M, d; \rho)_q$   to denote a $q$-ary code of length n,  size $M$, minimum
distance $d$, and covering radius $\rho$.

In \cite{vas}  Vasil'ev  used the switching construction to construct a family of binary nonlinear  perfect codes
with packing radius $e = 1$ and length $n = 2^m -1$ , $m \geq 4$.
In \cite{vas}  Vasil'ev  proved that there are more than   $q^{q^{cn}}$ nonequivalent such codes of length $n$,
for all sufficiently large $n$ and a constant $c = \frac{1}{2} - \varepsilon$.
In \cite{etz} Etzion and  Vardy used the switching construction to construct  binary  nonlinear perfect codes of full-rank.
In \cite{phe} Phelps and Villanueva used the switching construction to construct $q$-ary  nonlinear  perfect codes with ranks
of different sizes.

In this paper we use the switching construction to construct a family of $q$-ary   nonlinear quasi-perfect codes  with covering radius $\rho = 2$.
The codes we offer are of length $n = q^m$ and size $ M = q^{n - m - 1}$ where $q$ is a prime power, $q \geq 3$, $m$ is an integer, $m \geq 2$.
In this paper   we prove that there are more than   $q^{q^{cn}}$ nonequivalent such codes of length $n$, for all sufficiently large $n$ and a constant
$c = \frac{1}{q} - \varepsilon$.

In \cite{rom}, using the  concatenation construction, a family of $q$-ary even   quasi-perfect codes  with covering radius $\rho = 2$ was constructed.
The codes proposed in \cite{rom} have length  $n = q^m$ and size $ M = q^{n - m - 1}$ where $q$ is a prime power, $q \geq 3$, $m$ is an integer, $m \geq 4$.
It was shown   \cite{rom}  that there are more than   $q^{q^{cn}}$ nonequivalent such codes of length $n$,
for all sufficiently large $n$ and a constant $c$.

For $q = 2$, the construction from \cite{rom} is the concatenation construction of binary  extended   perfect codes proposed by Phelps in \cite{phe4}.

%%%%%%%%%%%%%%%%%%%%%%%%%%%%%%%%%%%%%%%%%%%%%%%%%%%
%%%%%%%%%%%%%%%%%%%%%%%%%%%%%%%%%%%%%%%%%%%%%%%%%%%

\section{Generalized Reed-Muller codes}\label{sec:gen}

%%%%%%%%%%%%%%%%%%%%%%%%%%%%%%%%%%%%%%%%%%%%%%%%%%%

Let $\mathbb{F}_{q}[X_1, X_2, \ldots, X_m]$ be the algebra of polynomials in $m$ variables over the field $\mathbb{F}_{q}$.
Let a  polynomial $f$ be in $\mathbb{F}_{q}[X_1, X_2, \ldots, X_m]$, then by $\text{deg}(f)$ we denote  its total degree.
Let $AG(m,q)$  be the $m$-dimensional affine space   over the field $\mathbb{F}_{q}$.
Let $n = q^m$ and the points $P_1, P_2, \ldots, P_n$ of $AG(m,q)$  be arranged in some fixed order.
Let  $r$ be an integer such that $0\leq r \leq (q-1)m$.
Then the generalized Reed-Muller code  \cite{kas} of order $r$  over the field $\mathbb{F}_{q}$ is the following subspace:
\begin{equation*}
{RM}_q(r,m) = \big\{ (f(P_1), f(P_2), \ldots, f(P_n)) \, \big| \,
f \in \mathbb{F}_{q}[X_1, X_2, \ldots, X_m], \text{deg}(f) \leq r \big\}.
\end{equation*}
The code ${RM}_q(r,m)$ has the following parameters (\cite{kas}, \cite[Theorem 5.5]{key}):
\begin{enumerate}
\item  the length  is  $q^m$;
\item the dimension is
\begin{equation}\label{eq1}
\sum_{k = 0}^{m}{(-1)}^k {m \choose k }{m + r - kq \choose r - kq };
\end{equation}
\item the minimum distance is
\begin{equation}\label{eq2}
(q - b)q^{m - a - 1},
\end{equation}
where $r = (q - 1)a + b$ and $ 0 \leq b < q - 1$.
\end{enumerate}

In what follows, we are only interested in the generalized  Reed-Muller codes ${RM}_q((r,m)$ of order $ r = (q - 1)m - 2$.
In the binary case, the code ${RM}_q((r,m)$ of order $ r = (q - 1)m - 2$ is  the extended Hamming code of length $ n = 2^m$.
The  dimension of the extended Hamming code of length $ n = 2^m$  is  $n - m - 1$   and the minimum distance is $4$.
\begin{proposition}
\label{Pr:1}
Let $q \geq 3$, $m \geq 1$, and $n = q^m$.
Then  the dimension of the generalized Reed-Muller code ${RM}_q((q-1)m - 2,m)$ is  ${n - m - 1}$
and the minimum distance is $3$.
\end{proposition}
\begin{proof}
If a code $\cal C$ belongs to the class of generalized Reed-Muller codes, then the dual code ${\cal C}^{\perp}$ also belongs to this class \cite{del}.
Assume that the length of a code $\cal C$ is  $n$.
Then
\begin{equation}\label{eq3}
\text{dim}({\cal  C}) + \text{dim}({\cal C}^{\perp}) = n.
\end{equation}
From \cite[Theorem 5.8]{key}, for $r < (q - 1)m$, we have
$${{RM}_q(r,m)}^{\perp} = {RM}_q((q-1)m - 1 - r,m).$$
Therefore, the order of the code dual to the code ${RM}_q((q-1)m - 2, m)$  is 1.
By formulas  (\ref{eq1}) and (\ref{eq3}), the dimension of the code ${RM}_q((q-1)m - 2,m)$ is   ${n - m - 1}$.

Now we show that for all $m \geq 1$  the minimum distance of   the generalized Reed-Muller code ${RM}_q((q-1)m - 2, m)$  is $3$.
Since
$$r = (q-1)m - 2 = (q - 1)a + b \,\,\, \text{and} \, \, \, 0 \leq b < q - 1, $$
for $q \geq 3$,  we have $a = m - 1$ and  $b = q - 3$.
Therefore, by the formula (\ref{eq2}), the minimum distance of   the generalized Reed-Muller code ${RM}_q((q-1)m - 2, m)$  is $3$.
\end{proof}

\begin{proposition}
\label{Pr:2}
The generalized Reed-Muller code ${RM}_q((q-1)m - 2, m)$  is a $q$-ary quasi-perfect code with covering radius $\rho = 2$.
\end{proposition}
\begin{proof}
The  first  order generalized Reed-Muller code is a linear two-weight code \cite{del}.
According to \cite[Theorem 5.10]{del2}, the code dual to a linear two-weight code has  covering radius  $\rho = 2$.
Therefore, the generalized Reed-Muller  codes of order $(q - 1)m - 2$ are  quasi-perfect codes with covering radius $\rho = 2$.
\end{proof}

From \cite[Theorem 3.11]{van} it follows that linear quasi-perfect codes are uniformly packed codes.
Uniformly packed codes attain the Johnson bound and are optimal codes \cite[Theorem 1.3]{van}.

%%%%%%%%%%%%%%%%%%%%%%%%%%%%%%%%%%%%%%%%%%%%%%%%%%%
%%%%%%%%%%%%%%%%%%%%%%%%%%%%%%%%%%%%%%%%%%%%%%%%%%%

\section{Switching construction}\label{sec:switch}

%%%%%%%%%%%%%%%%%%%%%%%%%%%%%%%%%%%%%%%%%%%%%%%%%%%

The parity-check matrix $H$ of  ${RM}_q((q-1)m - 2, m)$ is $ (m + 1)\times q^m$ matrix
that contains the all-unity vector of length $ n = q^m$ and all transposed  vectors
of $\mathbb{F}_{q}^m$, see \cite{del}.
Let $H = \left[{\bf h}_1, {\bf h}_2,  \ldots,  {\bf h}_n\right]$, where
${\bf h}_1, {\bf h}_2,  \ldots,  {\bf h}_n$ are columns of the matrix $H$.
The affine space $AG(m,q)$ will be regarded  as the incidence geometry.
Points $P_1, P_2, \ldots, P_n$ of the  affine space $AG(m,q)$ correspond
to the coordinates of the vector space $\mathbb{F}_{q}^{n}$.

Let ${\bf x} = (x_1, x_2, \dots, x_n) \in \mathbb{F}_{q}^{n}$,
then  the \emph{support} of the vector ${\bf x}$ is the set
$$supp({\bf x}) = \{ i \, \, | \, \, x_i \neq 0 \}.$$
By Proposition~\ref{Pr:1},  the minimum distance  of the  code ${RM}_q((q-1)m - 2,m)$ is $3$ if  $q \geq 3$.
A codeword of weight $3$ of the  code ${RM}_q((q-1)m - 2, m)$ we will call a \emph{triple}.
Let ${\bf c}  = (c_1, c_2, \ldots, c_n)$  be a triple and $supp({\bf c}) = \{i, j, k\}$. Then
\begin{enumerate}
\item
the triple ${\bf c}$ lies on a  line $l$ if  $\{P_i, P_j, P_k\} \subseteq l$;
\item
the corresponding columns ${\bf h}_i, {\bf h}_j, {\bf h}_k$ are linearly dependent, i.e.,
$$c_i{\bf h}_i + c_j{\bf h}_j + c_k{\bf h}_k = \bf 0.$$
\end{enumerate}

Let $i \in \{1, \dots, n \}$, then by ${\cal R}_i$ we denote a subspace spanned by the set of all triples of ${RM}_q((q-1)m - 2, m)$
having  $1$ in the  $i$th coordinate.
By definition,  the minimum distance of  ${\cal R}_i$ is 3.

\begin{proposition}
\label{Pr:3}
Let $q \geq 3$,  $m \geq 1$, and $n = q^m$. Let ${\cal R}_i \subseteq {RM}_q((q-1)m - 2, m)$. Then   for any $i \in \{1, 2, \ldots, n\}$
the dimension of ${\cal R}_i$  is $n - [m]_q - 1$,
where $[m]_q$ is the $q$-analogue of the natural number $m$.
\end{proposition}
\begin{proof}
For $q \geq 3$, every line of the affine space $AG(m,q)$  is in ${RM}_q((q-1)m - 2, m)$, see \cite{del,kasl}.
The generalized Reed-Muller code ${RM}_q((q-1)m - 2, m)$  is spanned by its minimum-weight vectors \cite{del,din}.
In an affine  space $AG(m,q)$, each point has  $\frac{n - 1}{q - 1}$ lines passing through it,
and each line contains $q$ points.
For every two distinct points, there is exactly one line that contains both points.
Thus, for each line, there are $q - 2$ linear independent triples that  lie on this line.
Therefore, the number of linear independent  triples that
generate ${\cal R}_i$ is $(q - 2) \frac{n - 1}{q - 1} = n - [m]_q - 1$.
\end{proof}

For $m = 1$, generalized Reed-Muller codes are  extended Reed-Solomon codes   \cite[p. 296]{mac}.
For $m = 1$, the affine space $AG(1,q)$  is a line,
and for ${\cal R}_i \subseteq {RM}_q(q - 3, 1)$,
$$\text{dim}({RM}_q(q - 3, 1)) = \text{dim}({\cal R}_i) = q - 2. $$

\begin{proposition}
\label{Pr:4}
Let $q \geq 3$, $m \geq 1$. Let ${\bf c}$ be a triple and    ${\bf c} \in {\cal R}_i \subseteq {RM}_q((q-1)m - 2, m)$. Then ${\bf c}$   lies on a  line passing through the point $P_i$.
\end{proposition}
\begin{proof}
If the triples lie on same line, then their linear combination lies on the same  line.
In an affine  space $AG(m,q)$, the intersection of any two distinct lines contains exactly one point.
Therefore, triples that lie on distinct lines passing through the point $P_i$ have only one common point $P_i$
and their linear combination contains more than three points.
\end{proof}
By Proposition~\ref{Pr:2}, the  code ${RM}_q((q-1)m - 2, m)$  is a $q$-ary quasi-perfect code with covering radius $\rho = 2$.
For a given  code ${RM}_q((q-1)m - 2, m)$ and $t = 0, 1, 2$, define
\begin{equation*}
{RM}_q^{(t)}((q-1)m - 2, m) = \big\{{\bf x} \in  \mathbb{F}_{q}^n  \,\, \big| \,\, d\big({\bf x},{RM}_q((q-1)m - 2, m)\big) = t \big\}.
\end{equation*}
The sets ${RM}_q^{(t)}((q-1)m - 2, m)$  are called the \emph{subconstituents} of  ${RM}_q((q-1)m - 2, m)$.
\begin{proposition}
\label{Pr:5}
Let $q \geq 3$, $m \geq 2$. Let ${\bf x} \in {RM}_q^{(2)}((q-1)m - 2, m)$, $wt({\bf x})= 2$, and $supp({\bf x}) = \{j,k\}$.
Let $i \neq j$, $i \neq k$, and the points $P_i, P_j, P_k$  are collinear.
Then there is a triple ${\bf c} \in {\cal R}_i \subset {RM}_q((q-1)m - 2, m)$ such that  $supp({\bf c}) = \{i,j,k\}$ and $d({\bf c}, {\bf x}) =2$.
\end{proposition}
\begin{proof}
Consider ${\bf x} = (x_1, x_2, \ldots, x_n) \in {RM}_q^{(2)}((q-1)m - 2, m)$, $n = q^m$.
Let $wt({\bf x})= 2$ and   $supp({\bf x}) = \{j,k\}$.
Then there are $q - 2$ linear independent triples that lie on a line  $l$ passing through the points $P_j$ and $P_k$.
Let $i \neq j$, $i \neq k$, and the points $P_i, P_j, P_k$  are collinear.
Then there is a triple ${\bf c} =  (c_1, c_2, \ldots, c_n) $ lies on the line $l$   such that $supp({\bf c}) = \{i,j,k\}$.
For some scalar $\lambda \in \mathbb{F}_{q}\setminus \{0\}$, either $x_j = \lambda c_j$ or $x_k = \lambda c_k$, and $d({\bf x},\lambda{\bf c}) = 2$.
\end{proof}

Let ${\bf e}_i$ denote the vector of length $n$ having all components equal to zero, except the $i$th component which
contains a one.
Let ${\bf x} \in{RM}_q((q-1)m - 2, m)$,  $\lambda \in \mathbb{F}_{q}\setminus \{0\}$.
Let ${\cal R}_i + {\bf x} + \lambda{\bf e}_i$  be a translate of ${\cal R}_i + {\bf x}$.
We shall define a \emph{switching} to be the process of the replacing the coset
${\cal R}_i + {\bf x}$ with the translate ${\cal R}_i + {\bf x} + \lambda{\bf e}_i$.

\begin{theorem}
\label{Pr:6}
Let $q\geq 3$, $m \geq 2$, and $n = q^m$. Let ${\cal R}_i + {\bf x} \subset {RM}_q((q-1)m - 2, m )$.
Then
$$ {\cal C}' = \Big( {RM}_q((q-1)m - 2, m ) \setminus ({\cal R}_i + {\bf x})\Big) \cup \Big( {\cal R}_i + {\bf x} + \lambda{\bf e}_i \Big)$$
is a $q$-ary nonlinear quasi-perfect code with  parameters $(n, q^{n - m - 1}, 3; 2)_q$,
for any  $\lambda \in \mathbb{F}_{q}\setminus \{0\}$ and for any  $i \in\{1, 2, \ldots, n\}$.
\end{theorem}
\begin{proof}
It is easy to show that ${\cal C}'$ is nonlinear. Let us show that $d({\cal C}') = 3$. Let ${\bf y} \in ({\cal R}_i + {\bf x})$ and
${\bf z} \in {RM}_q((q-1)m - 2, m ) \setminus ({\cal R}_i + {\bf x})$. Assume that $d({\bf y} + \lambda{\bf e}_i, {\bf z}) = 2$.
Then  $d(\lambda{\bf e}_i, {\bf z} - {\bf y}) = 2$.
Thus, ${\bf z} - {\bf y}$ is a triple that lies on a line passing through the point $P_i$.
So  $({\bf z} - {\bf y}) \in {\cal R}_i$  and ${\bf z} \in {\cal R}_i +  {\bf x}$.
We obtained a contradiction.
Therefore, the minimum distance of the code ${\cal C}'$ is 3.

Now let us show that  $\rho({\cal C}') = 2$.  Let ${\bf y} \in {RM}_q^{(1)}((q-1)m - 2, m)$.
In this case, it is obvious that there exists ${\bf c} \in {\cal C}'$ such that $d({\bf c}, {\bf y}) \leq 2$.

Let ${\bf y} \in {RM}_q^{(2)}((q-1)m - 2, m)$.
Then there exists ${\bf c} \in {RM}_q((q-1)m - 2, m)$ such that $d({\bf c}, {\bf y}) = 2$.
If ${\bf c} \in {RM}_q((q-1)m - 2, m ) \setminus ({\cal R}_i + {\bf x})$, then ${\bf c} \in {\cal C}'$.
Let ${\bf c} \in {\cal R}_i + {\bf x}$.
Since $wt({\bf y} -  {\bf c}) = 2$, let $supp({\bf y} -  {\bf c}) = \{j, k\}$.
\begin{enumerate}
\item If the points $P_i, P_j, P_k$  are not collinear,
then, by Proposition~\ref{Pr:4}, there exists  a triple ${\bf c}' \notin {\cal R}_i$
such that $d({\bf y},{\bf c}' +  {\bf c}) = 2$.
\item If the  points $P_i, P_j, P_k$  are  collinear.
If $i = j$ or  $i = k$, then it is obvious that there is a triple
${\bf c}' \in {\cal R}_i$  such that $d({\bf y},{\bf c}' +  {\bf c} + \lambda{\bf e}_i) = 2$.
If $i \neq j$,  $i \neq k$,
then, by Proposition~\ref{Pr:5}, there exists  a triple ${\bf c}' \in {\cal R}_i$
such that $d({\bf y},{\bf c}' +  {\bf c} + \lambda{\bf e}_i) = 2$.
\end{enumerate}
\end{proof}

\begin{theorem}
\label{Pr:7}
Let $q\geq 3$, $m \geq 2$, and $n = q^m$.
Let ${\cal R}_i + {\bf x}_t \subset {RM}_q((q-1)m - 2, m )$, $ 1\leq t \leq q^{[m]_q - \, m}$,
$( {\cal R}_i + {\bf x}_{t_1}) \cap ( {\cal R}_i + {\bf x}_{t_2}) = \varnothing$
for all \, $1\leq t_1 \leq q^{[m]_q - m}$,  $1\leq t_2 \leq q^{[m]_q - m}$,  $t_1 \neq t_2$.
Then
$$ {\cal C}' =  \bigcup_{t = 1}^{q^{[m]_q - m}} ( {\cal R}_i + {\bf x}_t + \lambda_t{\bf e}_i)$$
is a  $q$-ary  quasi-perfect code  with  parameters $(n, q^{n - m - 1}, 3; 2)_q$,
$\lambda_t \in \mathbb{F}_{q}$,  $i \in\{1, 2, \ldots, n\}$.
\end{theorem}
\begin{proof}
The number of cosets ${\cal R}_i + {\bf x}_t$ is $q^{[m]_q - m}$.
The cosets ${\cal R}_i + {\bf x}_t$  form a partition of ${RM}_q((q-1)m - 2, m)$.
Let us show that $d({\cal C}') = 3$.
Let ${\bf y} \in ({\cal R}_i + {\bf x}_{t_1} + \lambda_{t_1}{\bf e}_i)$,
${\bf z} \in  ({\cal R}_i + {\bf x}_{t_2} +  \lambda_{t_2}{\bf e}_i)$, and
$( {\cal R}_i + {\bf x}_{t_1}) \cap ( {\cal R}_i + {\bf x}_{t_2}) = \varnothing.$
Since
$$({\bf y } - \lambda_{t_2}{\bf e}_i) \in ({\cal R}_i + {\bf x}_{t_1} + (\lambda_{t_1} - \lambda_{t_2}){\bf e}_i), $$
$$({\bf z } - \lambda_{t_2}{\bf e}_i) \in ({\cal R}_i + {\bf x}_{t_2}) \subset {RM}_q((q-1)m - 2, m )  \setminus ({\cal R}_i + {\bf x}_{t_1}),$$
by Theorem~\ref{Pr:6} we have $d(({\bf y } - \lambda_{t_2}{\bf e}_i),({\bf z } - \lambda_{t_2}{\bf e}_i)) \geq 3.$

Now let us show that  $\rho({\cal C}') = 2$.
Let ${\bf y} \in {RM}_q^{(2)}((q-1)m - 2, m)$.
Then there exists ${\bf c} \in {RM}_q((q-1)m - 2, m)$ such that $d({\bf c}, {\bf y}) = 2$.
Let ${\bf c} \in {\cal R}_i + {\bf x}_t$.
Since $wt({\bf y} -  {\bf c}) = 2$, let $supp({\bf y} -  {\bf c}) = \{j, k\}$.
Assume that  the  points $P_i, P_j, P_k$  are  collinear.
If $i = j$ or  $i = k$, then it is obvious that there is a triple
${\bf c}' \in {\cal R}_i$  such that $d({\bf y},{\bf c}' +  {\bf c} + \lambda_t{\bf e}_i) = 2$
and ${\bf c}' +  {\bf c} \in {\cal R}_i + {\bf x}_t$ .
If $i \neq j$,  $i \neq k$,
then, by Proposition~\ref{Pr:5}, there exists  a triple ${\bf c}' \in {\cal R}_i$
such that $d({\bf y},{\bf c}' +  {\bf c} + \lambda_t{\bf e}_i) = 2$ and ${\bf c}' +  {\bf c} \in {\cal R}_i + {\bf x}_t$ .

If points $P_i, P_j, P_k$  are not collinear,
then, by Proposition~\ref{Pr:4}, there exists  a triple ${\bf c}' \notin {\cal R}_i$
such that $d({\bf y},{\bf c}' +  {\bf c}) = 2$ and ${\bf c}' +  {\bf c} \notin {\cal R}_i + {\bf x}_t $.
\end{proof}

%%%%%%%%%%%%%%%%%%%%%%%%%%%%%%%%%%%%%%%%%%%%%%%%%%%
%%%%%%%%%%%%%%%%%%%%%%%%%%%%%%%%%%%%%%%%%%%%%%%%%%%

\section{Lower bound}\label{sec:low}

%%%%%%%%%%%%%%%%%%%%%%%%%%%%%%%%%%%%%%%%%%%%%%%%%%%

\begin{theorem}
\label{Pr:8}
Let $q\geq 3$,  $m \geq 2$, and $n = q^m $. Then  there are more than
$q^{q^{[m]_q -  m} - \, n(m + 2)}$
nonequivalent $q$-ary  quasi-perfect codes with  parameters $(n, q^{n - m - 1}, 3; 2)_q$.
\end{theorem}
\begin{proof}
The number of words of length $q^{[m]_q -  m}$  over $\mathbb{F}_{q}$ is $q^{q^{[m]_q -  m}}$.
Hence, by Theorem~\ref{Pr:7},  the number of different quasi-perfect codes
with  parameters $(n, q^{n - m - 1}, 3; 2)_q$ is $q^{q^{[m]_q -  m}}$.
The number of monomial matrices   of size $n\times n$ over $\mathbb{F}_{q}$ is $(q - 1)^{n}n!$.
Any equivalence class contains not more than $ (q - 1)^{n}n! \cdot q^n$
different quasi-perfect codes with  parameters $(n, q^{n - m - 1}, 3; 2)_q$.
Therefore, there are more than
$q^{q^{[m]_q -  m} - \, n(m + 2)}$
nonequivalent $q$-ary quasi-perfect codes
with  parameters $(n, q^{n - m - 1}, 3; 2)_q$.
\end{proof}

\begin{corollary}
\label{Pr:9}
Let $q\geq 3$, $m \geq 2$, and $n = q^m $. Then  there are more than
$q^{q^{cn}}$
nonequivalent $q$-ary nonlinear quasi-perfect codes with  parameters $(n, q^{n - m - 1}, 3; 2)_q$
for all sufficiently large $n$ and a constant
$c = \frac{1}{q} - \varepsilon$.
\end{corollary}
%%%%%%%%%%%%%%%%%%%%%%%%%%%%%%%%%%%%%%%%%%%%%%%%%%%


\begin{thebibliography}{14}

\bibitem{key}
Assmus~E.~F.,  Key~J.~D.:
Polynomial codes and finite geometries.
In: Pless~V.~S.,  Huffman~W.~C., Brualdi~R.~A. (eds.)
Handbook of Coding Theory,
vol. II, pp. 1269--1344. Elsevier, Amsterdam   (1998).


\bibitem{del2}
Delsarte~P.:
An Algebraic Approach to the Association Schemes of Coding Theory.
{ Philips Res. Rep. Suppl.},
{\bf 10} (1973).

\bibitem{del}
Delsarte~P.,  Goethals~J.~M., MacWilliams~F.~J.:
On generalized Reed-Muller codes and their relatives.
{ Inform. Contr},
{\bf 16}, 403--442 (1970).

\bibitem{din}
Ding~P., Key~J.~D.:
Minimum-weight codewords as generators of generalized Reed-Muller codes.
{ IEEE Trans. Inf. Theory},
{\bf 46}, 2152--2157 (2000).

\bibitem{goe}
 Goethals~J.~M., van Tilborg~H.~C.~A.:
Uniformly Packed Codes.
{ Philips Res. Rep.},
{\bf 30}, 9--36 (1975).

\bibitem{etz}
 Etzion~T.,  Vardy~A.:
Perfect binary codes: Constructions, properties, and enumeration.
{ IEEE Trans. Inf. Theory},
{\bf 40}, 754--763 (1994).

\bibitem{kas}
Kasami~T., Lin~S.,  Peterson~W.~W.:
New generalizations of the Reed-Muller codes. Part I: Primitive codes.
{IEEE Trans. Inf. Theory},
{\bf 14}, 189--199 (1968).

\bibitem{kasl}
Kasami~T., Lin~S.,  Peterson~W.~W.:
Polynomial codes.
{IEEE Trans. Inf. Theory},
{\bf 14}, 807--814 (1968).

\bibitem{mac}
MacWilliams~F.~J., Sloane~N.~J.~A.:
{The Theory of Error-Correcting Codes}.
North-Holland Publishing Co.,  Amsterdam (1977).

\bibitem{phe4}
Phelps~K.~T.:
A general product construction for error correcting codes.
{SIAM J. Algebraic Discrete Methods},
{\bf 5}, 224--228 (1984).

\bibitem{phe}
Phelps~K.~T.,  Villanueva~M.:
Ranks of $q$-ary 1-perfect codes.
{ Des. Codes  Cryptogr.},
{\bf 27}, 139--144 (2002).

\bibitem{rom}
Romanov~A.~M.:
On perfect and Reed-Muller codes over finite fields.
Probl. Inf. Transm.,
{\bf 57}, 199--211 (2021).

\bibitem{van}
van Tilborg~H.~C.~A.:
{Uniformly Packed Codes}.
PhD Thesis, University of Tech., Eindhoven
(1976).

\bibitem{vas}
Vasil'ev~Yu.~L.:
On nongroup close-packed codes.
{ Probl. Kybern.},
{\bf 8}, 337--339 (1962).


\end{thebibliography}
\end{document}